\documentclass[12pt]{article}

\usepackage{amsmath,amssymb,amsthm,bm,cite,fullpage}
\usepackage{mathptmx,eucal} 
\usepackage{slashed} 

\renewcommand{\baselinestretch}{1.1}

\title{On the assertion that PCT violation implies Lorentz
non-invariance}

\author{Michael~D\"utsch\,\dag\ and Jos\'e M. Gracia-Bond\'ia\,\ddag
\\ \\
\dag\, Institut f\"ur Theoretische Physik, Universit\"at G\"ottingen
\\
Friedrich-Hund-Platz 1, 37077 G\"ottingen, Germany
\\ \\
\ddag\, Departamento de F\'isica Te\'orica, Universidad de Zaragoza
\\
50009 Zaragoza, Spain
\\
and
\\
Instituto de F\'{\i}sica Te\'orica, CSIC--UAM, Madrid 28049, Spain
}

\date{\today}

\newcommand{\Th}{\Theta}          

\newcommand{\sS}{\mathcal{S}}     

\newcommand{\bN}{\mathbb{N}}      
\newcommand{\bR}{\mathbb{R}}      

\newcommand{\nS}{\mathbf{S}}      
\newcommand{\nT}{\mathbf{T}}      

\DeclareMathOperator{\supp}{supp} 

\newcommand{\ovl}{\overline}      

\newcommand{\7}{\dagger}          

\newcommand{\vecform}{\bm}              

\newcommand{\xx}{\vecform{x}}           

\newcommand{\word}[1]{\quad\mbox{#1}\quad} 
\def\wick:#1:{\mathopen:#1\mathclose:} 

\def\scal<#1|#2>{\langle#1\mathbin|#2\rangle} 

\theoremstyle{plain}
\newtheorem{thm}{Theorem}           
\newtheorem{prop}[thm]{Proposition} 
\newtheorem{hyp}{Assumption}        

\theoremstyle{definition}
\newtheorem{remk}[thm]{Remark}      

\newcommand{\hideqed}{\renewcommand{\qed}{}} 

\DeclareRobustCommand{\qned}{
  \ifmmode
  \else \leavevmode\unskip\penalty9999 \hbox{}\nobreak\hfill
  \fi
  \quad\hbox{\qnedsymbol}}
\newcommand{\qnedsymbol}{$\boxminus$} 

\makeatletter
\renewcommand{\section}{\@startsection{section}{1}{\z@}%
                        {-3.5ex \@plus -1ex \@minus -.2ex}%
                        {2.3ex \@plus.2ex}%
                        {\normalfont\large\bfseries}}
\renewcommand{\subsection}{\@startsection{subsection}{2}{\z@}%
                        {-3.25ex \@plus -1ex \@minus -.2ex}%
                        {1.5ex \@plus .2ex}%
                        {\normalfont\normalsize\bfseries}}
\renewcommand{\@dotsep}{200} 
\makeatother

\begin{document}

\maketitle

\begin{abstract}
Out of conviction or expediency, some current research programs
\cite{Kbook,KR,FA,DA} take for granted that ``PCT violation implies
violation of Lorentz invariance''. We point out that this claim
\cite{colinamuyverde} is still on somewhat shaky ground. In fact, for
many years there has been no strengthening of the evidence in this
direction. However, using causal perturbation theory, we prove here
that when starting with a local PCT-invariant interaction, PCT
symmetry can be maintained in the process of renormalization.
\end{abstract}

\renewcommand{\baselinestretch}{1}


\section{Introduction}

In dealing with fundamental questions of science, it may be advisable
to take a cue from knowledgeable philosophers. A recent account on
PCT%
\footnote{Since we stand from \textit{causal perturbation theory}, we
do not plump for the nowadays popular~``CPT''.}
invariance by one such~\cite{gabriel} registers the fact that
arguments for PCT conservation in relativistic field theory fall into
two neatly separated classes. Heuristic treatments essentially amount
to observing that ``it does not appear possible to construct a
`reasonable' interaction which violates PCT''~\cite{casadeorange}.
Rigorous proofs were often based on the Wightman axioms%
\footnote{Among the treatments of Wightman theory, the traditional
ones \cite{StreaterW,Jost} remain the best, in our mind. With some
tweaking of the original axioms ---exemplified in \cite{himself} for
the Maxwell field--- this method does prove PCT symmetry for many
kinds of free fields.}
or the axioms of algebraic QFT ---the exceedingly beautiful theory
in~\cite{GL95} comes to mind--- so holding within their limited domain
of validity.

On this stark dichotomy we quote~\cite{gabriel}. ``We are thus faced
with the following dilemma: (A)~The Weinberg and textbook Lagrangian
formalisms are complete but typically mathematically ill-formed.
(B)~The axiomatic and algebraic formalisms are incomplete but
mathematically well-formed.'' \textit{Completeness} in context means
that non-trivial, realistic interacting quantum field models can be
formulated in the approach.

However, one paper in the literature presumes to bridge the chasm,
claiming in its title ``PCT violation implies violation of Lorentz
invariance". Beyond preambles, reaching the conclusion there takes a
grand total of 20~half-lines and \textit{one} displayed
formula~\cite{colinamuyverde}.
Such a \textit{blitzkrieg} might breed distrust. Ponder for instance the
classic proof by Epstein of PCT invariance of the $\nS$-matrix for
theories of local observables~\cite{E67}. It goes for eighteen tight
pages of complex and functional analysis. Within its framework it
remains state of the art ---consult~\cite{Arakiri}. At its end Epstein
declares ``\ldots it hardly needs to be remarked that the result is
not expected to strengthen the evidence for the PCT invariance of
nature''.

What about the literature on PCT conservation for \textit{particular}
interacting models in the realm of renormalized, perturbative QFT? On
the basis of the Glaser--Lehmann--Zimmermann~(GLZ) theorem
\cite{GreatOldMasters}, Steinmann undertook in~\cite{SteinmannEin} to
construct a Wightman-like perturbative expansion for the ostensibly
PCT-invariant self-interaction vertex of a neutral scalar field. After
formidable prerequisites, his proof of PCT symmetry for the model goes
on for more than ten~pages. A different tack was taken
in~\cite{Scharf}: perturbative QED is developed in terms of
time-ordered products (TOP), constructed by means of causal
renormalization in the manner of Epstein and Glaser~\cite{EpsteinG}.
Painstakingly as well, Scharf there manages to show that TOP can be
forged in a PCT invariant way. (His method is arguably valid for any
model involving P-, C- and T-invariant vertices to begin~with.)

Surely that was enough for many an expert not to bother with
\cite{colinamuyverde}. However, the stakes have become higher of late.
For an incautious reader of \cite{colinamuyverde}, a failure of PCT
conservation in nature would lead ``beyond special relativity''
automatically. Since vast current research programs on possible
violation of Lorentz and PCT symmetries in nature appear to assume
this, to keep ignoring the issue will not do. As well, experimental
results as diverse as~\cite{enlafrente} and \cite{enlaboca},%
\footnote{Long-standing tensions in the neutrino sector motivated the
path-breaking paper \cite{BL}. Now the tension has been relieved by
the update~\cite{delosultimosdias}.}
where Lorentz invariance and symptoms of apparent PCT violation seem
to coexist, must give us pause.

In the next section of this letter we try to puzzle out the argument
of \cite{colinamuyverde}. We point out that the assumptions needed to
apply it have not been verified for any non-trivial, realistic model.
In Section~3 we review the treatment of interacting fields in the
framework of causal perturbation theory. This prepares the ground for
Section~4, where we strengthen the perturbative treatment of PCT
symmetry by a different path altogether. Section~5 contains our final
situation assessment.

\section{A bridge too far}

An interacting QFT model is defined in \cite{colinamuyverde} as
Lorentz invariant if the ``$\tau$-functions'' (vacuum expectation
values of TOP) are Lorentz covariant. We do not take exception to
this. The $\tau$-functions are exhibited there as
\begin{equation}
\tau^{(n)}(x_1,\ldots,x_n) := \sum_P H\big(t_{P1},\ldots,t_{Pn}\big)
W^{(n)}(x_{P1},\ldots,x_{Pn});
\label{eq:quosque-perveniatur}
\end{equation}
where $x_i=(t_i,\xx_i)$, the sum is over permutations $P$ of $n$
points, $W^{(n)}$ denotes putative Wightman functions, and $H$ is the
Heaviside function which enforces $t_{P1}\ge\cdots\ge t_{Pn}$. This is
\textit{the} formula. Then comes the punch~line: the $W$-functions are
to be shown to be \textit{weakly local}, so PCT must hold.

\smallskip

Bold identities like \eqref{eq:quosque-perveniatur} are haunted by the
question of existence. In plain language: it is not quite clear what
either side of the formula means. Helpfully, aside from vacuum
expectation values having been taken, we notice that
\eqref{eq:quosque-perveniatur} is identical in form to~(38) in
\cite{EpsteinG}. There the fields being time-ordered are Wick
polynomials of \textit{free} (incoming) fields, and the formula is
given as a \textit{tentative definition} of its left hand side. As
such it is nothing but the ``solution'' for TOP in terms of
unrenormalized Feynman graphs. Epstein and Glaser hasten to indicate
that expressions like the right hand side in
\eqref{eq:quosque-perveniatur} are illegitimate for $n>2$, since the
$W^{(n)}$ are distributions, whose product with the Heaviside
functions is undefined: this is merely an instance of the ultraviolet
problem of perturbation theory. On the face of it equation
\eqref{eq:quosque-perveniatur}, as it appears in
\cite{colinamuyverde}, ignores the need for renormalization at its
peril.

Of course, the author of \cite{colinamuyverde} is referring to
$W$-functions for interacting fields. The issue of bad definition of
\eqref{eq:quosque-perveniatur} still stands. Please bear with us, as
we attempt to salvage it. It may be argued that, whatever the messy
avatars of a proper construction for the
left hand side,%
\footnote{It may involve new couplings and even new fields
\cite{Aurora}.}
it is to respect~\eqref{eq:quosque-perveniatur} insofar as $t_i\ne
t_k$ for all $i<k$ (note that asking solely for $x_i\ne x_k$ for all
$i<k$ would be insufficient). So begin with any point
$(x_1,\ldots,x_n)$ such that the sum of linear combinations of the
differences $x_i-x_{i-1}$ with non-negative coefficients (one at least
being nonzero) is space-like, moreover fulfilling
$t_1>t_2>\cdots>t_n$; such points do exist. One may
use~\eqref{eq:quosque-perveniatur} there. Then perform a Lorentz
transformation such that the transformed point $(x'_1, \ldots,x'_n)$
satisfies $t'_n>\cdots>t'_2>t'_1$; such Lorentz transformations exist.
One may again enforce \eqref{eq:quosque-perveniatur} for
$(x'_1,\ldots,x'_n)$. Hence Lorentz invariance of the $\tau$- and
$W$-functions implies weak local commutativity for such points
$(x_1,\ldots,x_n)$:
\begin{equation*}
W^{(n)}(x_1,\ldots,x_n) = W^{(n)}(x_n,\ldots,x_1).
\end{equation*}
This suffices for PCT invariance of the $W$-functions and PCT
covariance of the fields, by a well-trodden argument by
Jost~\cite{Jost}, \textit{provided} that Wightman's axiom 0 (about the
state space), axiom~I (about the domain and continuity of the fields)
and axiom II (about the Lorentz covariance of the fields) are verified
---we borrow the numbering of the axioms from \cite{StreaterW}.
Namely, Jost's proof crucially uses the analyticity properties of the
$W$-functions, which follow from these axioms. The fact that Greenberg
mentions this kind of additional assumptions only in a footnote to
\cite{colinamuyverde} may have caused misinterpretations of his
statement.

Obviously the discussion in \cite{colinamuyverde} may be relevant
only for interacting models satisfying the Wightman axioms except {\it
local commutativity}, that is to say
\begin{equation}
[\phi(x),\,\phi(y)] = 0 \word{for} (x-y)^2 < 0,
\label{eq:0.6}
\end{equation}
for the case of a boson field $\phi$ (this is Wightman's axiom III).
To wit, local commutativity implies trivially weak local
commutativity, and hence for a model satisfying \eqref{eq:0.6} one may
dispense with the above contortions, since one immediately faces the
question of applicability of Jost's proof.%
\footnote{Incidentally, if Lorentz covariance is assumed on the level
of {\it operators} (that is, for the fields and their TOP) and not
only on the level of vacuum expectation values, local commutativity
can be derived very easily. For the benefit of the reader this is done
in an appendix.}

If we were dealing with Wick polynomials of free fields and their $W$-
and $\tau$-functions, for which the Wightman axioms hold true, all we
would obtain from the argument in~\cite{colinamuyverde} is PCT
invariance of a \textit{free} theory. Now, to the best of our
knowledge, for non-trivial realistic models one cannot ascertain
analyticity of Wightman-like functions; hence the argument \textit{\`a
la Jost} in~\cite{colinamuyverde} flounders. While the assertion that
PCT conservation holds for everyday interacting relativistic theories
remains plausible, to the question whether it has been proved at the
required level of rigour, the clear and present answer is: only for a
class of models ---for instance QED in \cite{Scharf} as above said---
and for none by Greenberg's argument.

\section{On interacting fields in causal perturbation theory}

To deal seriously with interacting $W$-functions, we naturally have
recourse to a rigorous theoretical framework. Causal perturbation
theory by Epstein and Glaser ``is closest to the spirit of Wightman's
axioms''~\cite{PS} and well suited for our purpose. There are three
steps. It entails first constructing the TOP of Wick polynomials;
second, employing them to derive the interacting fields; third,
performing the adiabatic limit (if available).

We sketch now the necessary detour. The Epstein--Glaser procedure was
developed on the footsteps of St\"uckelberg~\cite{HA49},
Bogoliubov~\cite{BS} and Nishijima~\cite{PR61}. Let $W$ be the vector
space of \textit{local} Wick polynomials $A(x)$. The TOP $T_n$ are
(multi)linear, totally symmetric maps from $W^{\otimes n}$ into the
space of operator-valued tempered distributions,
\begin{equation*}
T_n\big(A_1(x_1) \cdots A_n(x_n)\big) \in \sS'\big(\bR^{4n}\big),
\end{equation*}
satisfying the Bogoliubov--Shirkov--Epstein--Glaser axioms
\cite{EpsteinG,BS,DF04}. Besides Lorentz invariance, there is mainly
causality:
\begin{equation}
T_n\big(A_1(x_1) \cdots A_n(x_n)\big) = T_l\big(A_1(x_1) \cdots 
A_l(x_l)\big) \, T_{n-l}\big(A_{l+1}(x_{l+1}) \cdots A_n(x_n)\big),
\label{eq:3}
\end{equation}
if $\{x_1,...,x_l\}\cap\big(\{x_{l+1},\ldots,x_n\}+\ovl
V_-\big)=\emptyset$. We use generating functionals of the form
\begin{equation}
\nT\big(e^{A(h)}\big) := {\bf1} + \sum_{n=1}^\infty \, \frac1{n!} \int
dx_1\,\cdots\,dx_n \, T_n\big(A(x_1) \cdots A(x_n)\big) \,
h(x_1)\cdots h(x_n);
\label{eq:5}
\end{equation}
in particular the $\nS$-matrix:
\begin{equation}
\nS(g) := \nT\big(e^{i\,V(g)}\big), \quad g \in \sS(\bR^{4}),
\label{eq:9}
\end{equation}
for $V\in W$ a suitable first-order interaction Lagrangian.

Coupling constants have been replaced by Schwartz ``switching
functions'' collectively denoted by~$g$; this aims to excise the
infrared problem while dealing with the ultraviolet one. The
$\nS$-matrix operator on the Fock space of the incoming fields as a
functional of~$g$ is a centrepiece of the theory. It acts as a
generating function for the interacting fields by the
Stepanov--Polivanov--Bogoliubov formula~\cite{BS}. For the retarded
ones:
\begin{equation}
A_{V(g)}(x)\equiv A_\mathrm{ret}(g;x) := -i\,\frac{\delta}{\delta
h(x)} \bigg\vert_{h=0}\, \nS(g)^{-1} \, \nT\big(e^{i\,(V(g) + A(h))}
\big) \word{for} A\in W.
\label{eq:peregrinus-ubique}
\end{equation}
Both notations will be used. TOPs of these (retarded) interacting
fields may as well be defined in an analogous way by higher
derivatives with respect to $h$ ---see formulas (75) and (76) in
\cite{EpsteinG}. They are all functionals of $g$. Retarded interacting
fields are causal in the sense that
\begin{equation}\label{caus-supp}
A_{V(g) + V_1(g_1)}(x) = A_{V(g)}(x) \word{if} \supp g_1\cap(x + \ovl
V_-) = \emptyset\,;
\end{equation}
where $V_1\in W$ is arbitrary. With an obvious (re)normalization of
the TOP, the interacting fields in causal perturbation theory satisfy
the Yang--Feldman--K\"all\'en equations.%
\footnote{In the exact context of a well behaved background field this
was verified in~\cite{bellicoso}. A full treatment of interacting
fields for QED in causal perturbation theory was given
in~\cite{NC90}.}

\smallskip

{}From the causally renormalized interacting fields the interacting
$W$-functions are still some way off. One has to perform the
adiabatic limit as $g\uparrow1$, bristling with difficulties
\cite{EpsteinG,EpsteinGAs,PhilippeRoland}.%
\footnote{We mention that obstructions to the adiabatic limit are
discussed in a recent paper, invoking Fedosov's index to relate causal
renormalization and Wightman theory via Haag's local algebraic
formalism~\cite{dulcecracovia}.}
 
The good news is that local commutativity of the interacting fields
can be proved rather straightforwardly prior to the adiabatic limit.
It follows from the GLZ relation \cite{GreatOldMasters},
\begin{equation}
i\,\big[A_{V(g)}(x)\,,\,B_{V(g)}(y)\big] = \frac{\delta}{\delta h(y)}
\bigg\vert_{h=0}\,A_{V(g)+B(h)}(x)-\frac{\delta}{\delta h(x)}
\bigg\vert_{h=0}\,B_{V(g)+A(h)}(y)\quad (A,B\in W),
\end{equation} 
and the causality of the retarded interacting fields
\eqref{caus-supp}. In the present procedure the GLZ relation is a
consequence of the definition \eqref{eq:peregrinus-ubique} of the
interacting fields ---see Proposition 2 in \cite{DF99}. Alternatively
it can be taken as a defining axiom for perturbative interacting
fields~\cite{SteinmannEin,DF99}. Another piece of good news is that,
provided that the infrared behaviour is good, one can work with the
incoming Fock vacuum. In~\cite{EpsteinG,EpsteinGAs} Epstein and Glaser
were able to show that Wightman-like functions exist for purely
massive, asymptotically complete models.

Now the bad news. There's the rub: as far as we know, the non-linear
Wightman conditions have not been proved to hold in the
Epstein--Glaser formalism (where fields and the state space itself are
constructed as formal power series). It looks like a fearsome task,
and it may well happen that imposing too good a behaviour leads back
to overly strong restrictions on the nature of interaction. The matter
of applicability of Jost's line of proof for PCT invariance in the
causal framework is undecided as yet.

\section{PCT invariance survives renormalization}

Since the underlying issue is renormalization, turning from
\textit{blitzkrieg} to humble trench warfare, we expound a pertinent
result.

\smallskip
  
\begin{hyp}
The free fields are PCT-covariant, that is, there exists an
anti-unitary operator $\Theta$ in the Fock space of free fields such
that
$$
\Theta\Omega =\Omega\quad\text{and}\quad
\Theta\,\Phi(x)\,\Theta^{-1} = \Phi^c(-x),
$$
where $\Omega$ is the Fock vacuum and
$\Phi^c$ is a suitable conjugate of the field $\Phi$ (passing to the
adjoint and multiplication by a suitable matrix). Say, for a charged
scalar field $\phi$ or for a Dirac field $\psi$:
\begin{equation*}
\Theta\,\phi(x)\,\Theta^{-1} = \phi^\7(-x), \qquad
\Theta\,\psi(x)\,\Theta^{-1} = -i\,\gamma_5\,{\psi^\7}^T(-x).
\end{equation*}
The PCT transformation is not usually an involution; for fermions it
holds $\Theta^2=(-1)^F$, where $F$ is the number operator of the
fields, which implies $\Theta^2\,\psi(x)\,\Theta^{-2}=-\psi(x)$ and
$\Theta^4={\bf1}$. We will simply assume that
\begin{equation}
\Theta^{2N} = {\bf1} \word {for some} N\in \bN \setminus\{0\}.
\label{eq:1}
\end{equation}
\end{hyp}

\begin{prop}
For hermitian or charged scalar fields and for Dirac fields, normal
ordering commutes with the PCT transformation, that is normally
ordered products of free fields are also PCT-covariant:
\begin{equation}
\Theta\,\wick:\Phi(x_1) \cdots \Phi(x_n):\,\Theta^{-1} =
\wick:\Theta\,\Phi(x_1) \cdots \Phi(x_n)\,\Theta^{-1}:=
\wick:\Phi^c(-x_1) \cdots \Phi^c(-x_n):.
\label{eq:2}
\end{equation}
\end{prop}

\begin{proof}
Proceeding by induction on $n$, we use Wick's theorem in the form
\begin{align}
\wick:\Phi(x_1) \cdots \Phi(x_n): &= \wick:\Phi(x_1) \cdots
\Phi(x_{n-1}):\,\,\Phi(x_n)
\notag \\
&-\sum_{l=1}^{n-1}(\Omega,\Phi(x_l)\Phi(x_n)\,\Omega)\,\,
\wick:\Phi(x_1) \cdots \widehat{\Phi(x_l)} \cdots \Phi(x_{n-1}):.
\label{wick}
\end{align}
The notation means that $\Phi(x_l)$ is omitted. Since $\Theta$ is
anti-unitary, the induction assumption yields
\begin{align*}
&\Theta\,\wick:\Phi(x_1) \cdots \Phi(x_n):\,\Theta^{-1} =
\wick:\Phi^c(-x_1) \cdots \Phi^c(-x_{n-1}):\,\,\Phi^c(-x_n)
\\
&\quad -\sum_{l=1}^{n-1}(\Omega,\Phi(x_l)\Phi(x_n)\,\Omega)^*\,\,
\wick:\Phi^c(-x_1) \cdots \widehat{\Phi^c(-x_l)} \cdots
\Phi^c(-x_{n-1}): = \wick:\Phi^c(-x_1) \cdots \Phi^c(-x_n):,
\end{align*}
where we use that
\begin{equation*}
(\Omega\,,\,\Phi(x_l)\Phi(x_n)\,\Omega)^*=
(\Theta\Omega\,,\Theta\,\Phi(x_l)\,\Theta^{-1}\Theta\,\Phi(x_n)\,
\Theta^{-1}\Theta\,\Omega)
=(\Omega\,,\,\Phi^c(-x_l)\Phi^c(-x_n)\,\Omega)
\end{equation*}
and that \eqref{wick} holds also for the field $\Phi^c(-x)$.
\end{proof}

In this section we only study interactions $V$ which are local Wick
polynomials and scalar with respect to Lorentz transformations. One
additionally needs that $V$ be real: $V^\7(x)=V(x)$ on a dense
subspace. In various cases the above proposition implies that $V$ is
PCT-invariant
\begin{equation}
\Theta\,V(x)\,\Theta^{-1} = V(-x).
\label{eq:6}
\end{equation}
If $V$ is built from different kinds of free fields, it is perhaps not
entirely clear that the mentioned conditions on $V$ suffice for
\eqref{eq:6} to hold. Therefore, we take it as an assumption.

\begin{hyp}
The interaction $V$ is a local Wick polynomial $V\in W$, which is
PCT-invariant: $\Theta\,V(x)\,\Theta^{-1} = V(-x)\ .$
\end{hyp}

We turn to the PCT-transformation of TOP. Since PCT contains time
reversal, we need to introduce the {\it antichronological products}
$\big(\ovl T_n\big)_{n\in\bN}$. A sequence $(T_n)_{n\in\bN}$ of TOP
determines a pertinent sequence $\big(\ovl T_n\big)_{n\in\bN}$; the
$\ovl T_n$ are also multilinear and totally symmetric maps defined by
\begin{equation}
\ovl\nT\big(e^{-A(h)}\big): = \nT\big(e^{A(h)}\big)^{-1},
\label{eq:4}
\end{equation}
with $h\in\sS(\bR^4)$. In keeping with the previous notation,
\begin{equation*}
\ovl\nT\big(e^{-A(h)}\big) = {\bf1} + \sum_{n=1}^\infty \,
\frac{(-1)^n}{n!} \int dx_1\,\cdots\,dx_n \, \, \ovl T_n \big(A(x_1)
\cdots A(x_n)\big) \, h(x_1) \cdots h(x_n).
\end{equation*}
As the term ``antichronological'' indicates, the $\ovl T_n$ satisfy
\eqref{eq:3}, with the $\ovl T$-products on the right hand side in
\textit{reverse order}.

\begin{thm}
(a) Let $A^c(-x):=\Th A(x)\Th^{-1}$. The time ordered products $T_n$
can be (re)normalized in such a way that
\begin{align}
\Theta\,T_n\big(A_1(x_1) \cdots A_n(x_n)\big)\,\Theta^{-1} &= \ovl
T_n\big(A^c_1(-x_1) \cdots A^c_n(-x_n)\big),
\label{eq:7} \\
\Theta\,\ovl T_n\big(A_1(x_1) \cdots A_n(x_n)\big)\,\Theta^{-1}&=
T_n\big(A^c_1(-x_1) \cdots A^c_n(-x_n)\big),
\label{eq:8}
\end{align}
for arbitrary $A_i\in W$. That is to say, conjugation of $T_n$ and
$\ovl T_n$ by the PCT operator amounts to mutual exchange and
conjugation of the arguments.

(b) The $\nS$-matrix is PCT covariant:
\begin{equation}
\Theta\,\nS(g)\,\Theta^{-1} = \nS(\hat g)^{-1} \word{where} \hat g(x)
:= g(-x).
\label{10}
\end{equation}

(c) Advanced interacting fields
\begin{equation}
\label{eq:12}
A_\mathrm{adv}(g;x) := -i\,\frac{\delta}{\delta h(x)}\bigg\vert_{h=0}\,
\nT(e^{i\,(V(g) + A(h))})\,\nS(g)^{-1},
\end{equation}
are mapped by the PCT transformation into retarded interacting fields
\eqref{eq:peregrinus-ubique} and viceversa:
\begin{align}
\Theta\,A_\mathrm{adv}(g;x)\,\Theta^{-1} &= A^c_\mathrm{ret}(\hat
g;-x);
\notag \\
\Theta\,A_\mathrm{ret}(g;x)\,\Theta^{-1} &= A^c_\mathrm{adv}(\hat
g;-x).
\label{eq:13}
\end{align}
\end{thm}

\begin{remk}
Advanced interacting fields are ``anti-causal'' in the sense that
\begin{equation}\label{anticaus-supp}
A_{\rm adv}(g + g_1;x) = A_{\rm adv}(g;x) \word{if} \supp g_1\cap(x +
\ovl V_+) = \emptyset\ .
\end{equation}
The support properties \eqref{caus-supp} and \eqref{anticaus-supp} of
the retarded, respectively advanced interacting fields are consistent
with~\eqref{eq:13}.
\end{remk}

\begin{proof}
a) $\Longrightarrow$ (b): This is obtained straightforwardly by using
definitions \eqref{eq:5}, \eqref{eq:9}, \eqref{eq:4}, anti-linearity
of $\Theta$ and \eqref{eq:6}.

(a) and (b) $\Longrightarrow$ (c): Due to
\begin{equation*}
\frac{\delta}{\delta h(x)}\bigg\vert_{h=0}\,\nT\big(e^{i\,(V(g) +
A(h))}\big) \,\ovl\nT\big(e^{-i\,(V(g) + A(h))}\big) = 0,
\end{equation*}
the advanced field \eqref{eq:12} may alternatively be written as
\begin{equation*}
A_\mathrm{adv}(g;x) = i\,\frac{\delta}{\delta h(x)}\bigg\vert_{h=0}\,
\nS(g)\,\ovl\nT\big(e^{-i\,(V(g) + A(h))}\big).
\end{equation*}
With that and with (a), (b) and anti-linearity of $\Theta$ we obtain
\begin{align*}
\Theta\,A_\mathrm{ret}(g;x)\,\Theta^{-1} = i\,\frac{\delta}{\delta
h(x)} \bigg\vert_{h=0}\,\nS(\hat g)\,\ovl \nT\big(e^{-i\,(V(\hat g) +
A^c(\hat h))}\big) = A^c_{\rm adv}(\hat g;-x)\notag\ .
\end{align*}
The second relation in \eqref{eq:13} is proved analogously.

The proof of (a) goes by induction on $n$, following the
Epstein--Glaser construction. In contrast to the latter, we do not use
the distribution-splitting method; instead we borrow Stora's extension
of distributions method ---see \cite{PS,BF,BergB}. This shortens the
discussion.

The case $n=1$ follows from $T_1(A(x)):=A(x)=:\ovl T_1(A(x))$.

Going to the inductive step $n-1\mapsto n$, the causality
condition \eqref{eq:3} determines
\begin{equation}
T^\circ_n\big(A_1(x_1) \cdots A_n(x_n)\big) := T_n(A_1(x_1) \cdots
A_n(x_n)\big) \big\vert_{\sS(\bR^{4n} \setminus \Delta_n)},
\label{eq:16}
\end{equation}
where $\Delta_n:=\{(x_1,\ldots,x_n)\in \bR^{4n}\,|\,x_1=\cdots=x_n\}$,
uniquely in terms of the lower orders $(T_{l<n})$. Renormalization of
subgraphs is taken up by the inductive procedure. Since the
$(T_{l<n})$ are PCT-covariant \eqref{eq:7} by assumption, we obtain
\begin{align}
&\Theta\,T^\circ_n\big(A_1(x_1) \cdots A_n(x_n)\big)\,\Theta^{-1} =
\Th \, T_l\big(A_1(x_1) \cdots A_l(x_l)\big)T_{n-l}\big(A_1(x_{l+1})
\cdots A_n(x_n)\big)\Th^{-1}
\notag \\
&=\; \ovl T_l\big(A^c_1(-x_1) \cdots A^c_l(-x_l)\big)\,\ovl
T_{n-l}\big(A^c_{l+1}(-x_{l+1}) \cdots A^c_n(-x_n)\big) = \ovl
T^\circ_n\big(A^c_1(-x_1) \cdots A^c_n(-x_n)\big),
\label{eq:17}
\end{align}
if $\{x_1,...,x_l\}\cap\big(\{x_{l+1},...,x_n\}+\ovl V_-\big)=
\emptyset$, where $\ovl T^\circ$ is defined like in \eqref{eq:16} and
in the last equality it is used that $\ovl T^\circ$ factorizes
antichronologically. In the same way one derives \eqref{eq:8} for
$\Theta\,\ovl T^\circ_n\, \Theta^{-1}$.

It follows that PCT invariance \eqref{eq:7} and \eqref{eq:8} can be
violated only in the extension
\begin{equation}
\sS'\big(\bR^{4n} \setminus \Delta_n\big) \ni
T^\circ_n\big(A_1(x_1)\cdots\big) \to T_n\big(A_1(x_1)\cdots\big) \in
\sS'\big(\bR^{4n}\big),
\end{equation}
that is, in the process of renormalization.  To obtain a PCT-covariant
renormalization we take an arbitrary extension $T_n\big(A_1(x_1)
\cdots\big)\in \sS'\big(\bR^{4n}\big)$ of $T^\circ_n\big(A_1(x_1)
\cdots\big)$ fulfilling all other renormalization conditions (e.g.
Poincar\'e covariance, unitarity, power counting) and symmetrize it
with respect to the {\it finite} group generated by $\Theta$
\cite[App.~D]{DF04}:
\begin{align}
&T_n^\mathrm{sym}\big(A_1(x_1) \cdots A_n(x_n)\big) :=
\frac{1}{2N}\sum_{l=0}^{N-1}\Bigl(\Theta^{2l}\,
T_n\big(\Th^{-2l}A_1(x_1)\Th^{2l} \cdots \Th^{-2l}A_n(x_n)\Th^{2l}
\big)\,\Theta^{-2l}
\notag \\
&\quad +\Theta^{2l+1}\, \ovl T_n\big(\Th^{-2l-1}A_1(x_1)\Th^{2l+1}
\cdots \Th^{-2l-1} A_n(x_n)\Th^{2l+1}\big)\,\Theta^{-(2l+1)}\Bigr).
\label{eq:18}
\end{align}
This is also an extension of $T^\circ_n\big(A_1(x_1)\cdots\big)$, since
$\Theta^{2l}\,T_n(\cdots)\,\Theta^{-2l}$ and $\Theta^{2l+1}\,\ovl
T_n(\cdots)\,\Theta^{-(2l+1)}$ are respectively extensions of
$\Theta^{2l}\,T^\circ_n(\cdots)\,\Theta^{-2l}$ and $\Theta^{2l+1}\,\ovl
T^\circ_n(\cdots)\,\Theta^{-(2l+1)}$, and since
\begin{align*}
&\Theta^{2l}\,T_n^\circ\big(\Th^{-2l}A_1(x_1)\Th^{2l}\cdots\big)\,
\Theta^{-2l} = T_n^\circ\big(A_1(x_1)\cdots\big), \word{as well as}
\\
&\; \Theta^{2l+1}\,\ovl T_n^\circ\big(\Th^{-2l-1}A_1(x_1)
\Th^{2l+1}\cdots\big)\,\Theta^{-(2l+1)} = T_n^\circ \big(
A_1(x_1)\cdots\big),
\end{align*}
due to \eqref{eq:17}.

The only tricky part remaining is to show that the antichronological
product $\ovl T^\mathrm{sym}_n$ corresponding to $T^\mathrm{sym}_n$
according to its definition~\eqref{eq:4} can be written similarly as
\begin{align}
&\ovl T_n^\mathrm{sym}\big(A_1(x_1) \cdots A_n(x_n)\big) =
\frac{1}{2N}\sum_{l=0}^{N-1}\Bigl(\Theta^{2l}\, \ovl T_n\big(
\Th^{-2l}A_1(x_1) \Th^{2l} \cdots
\Th^{-2l}A_n(x_n)\Th^{2l}\big)\,\Theta^{-2l} \notag
\\
&\quad +\Theta^{2l+1}\,T_n\big(\Th^{-2l-1}A_1(x_1)\Th^{2l+1}
\cdots \Th^{-2l-1}A_n(x_n)\Th^{2l+1}\big)\,\Theta^{-(2l+1)}
\Bigr).
\label{eq:19}
\end{align}
Indeed, using this and \eqref{eq:1}, one sees that the
$T^\mathrm{sym}_n$ and the corresponding $\ovl T^\mathrm{sym}_n$
fulfil the assertions~\eqref{eq:7} and \eqref{eq:8}.

To prove \eqref{eq:19} we use that any extension $T_n$ and the
corresponding $\ovl T_n$ \eqref{eq:4} satisfy
\begin{align}
\big[T_n+ (-1)^n\ovl T_n\big]&\big(A_i(x_i)_{i\in\ovl n}\big) =
\sum_{M\subset\ovl n,\,1\le |M|\le n-1} (-1)^{|M|+1} \, \ovl
T_{|M|}\big(A_i(x_i)_{i\in M}\big)\,T_{n-|M|} \big(A_j(x_j)_{j\in\ovl
n\setminus M}\big)
\notag \\
&= \sum_{M\subset\ovl n\,,\,1\leq |M|\leq n-1}
(-1)^{n-|M|+1} T_{|M|}\big(A_i(x_i)_{i\in M}\big)\,\ovl T_{n-|M|}
\big(A_j(x_j)_{j\in\ovl  n\setminus M}\big),
\label{eq:20}
\end{align}
where $\ovl n:=\{1,...,n\}$ and $|M|$ is the size of block $M$; these
are just the relations $\ovl\nT(e^{-A(h)}\big) \nT\big(e^{A(h)}
\big)={\bf1}=\nT\big(e^{A(h)}\big)\ovl \nT\big(e^{-A(h)}\big)$. Since
the two expressions on the right hand side of \eqref{eq:20} are
inductively given, the proof is complete if we succeed to show that
expression \eqref{eq:19} for $\ovl T^\mathrm{sym}_n$ fulfils
\begin{equation}
T^\mathrm{sym}_n+(-1)^{n}\ovl T^\mathrm{sym}_n = T_n + (-1)^n\ovl 
T_n\,,
\label{eq:21}
\end{equation}
for the arbitrary extension $T_n$ used in \eqref{eq:18}. From
\eqref{eq:20} and PCT invariance of the lower orders $(T_{l<n})$ and
$(\ovl T_{l<n})$ we obtain
\begin{align*}
&\Theta\,\big[\ovl T_n + (-1)^nT_n\big]\big(A_i(x_i)_{i\in\ovl
n}\big) \, \Theta^{-1}
\\
&\quad = \sum_{M\subset\ovl n\,,\,1\leq |M|\leq n-1} (-1)^{n-|M|+1}
T_{|M|}\big(\Th\,A_i(x_i)_{i\in M}\,\Th^{-1}\big)\,\ovl
T_{n-|M|}\big(\Th \,A_j(x_j)_{j\in\ovl n\setminus M}\,\Th^{-1}\big)
\\
&\quad = \big[T_n + (-1)^n\ovl T_n\big]\big(\Th\,A_i(x_i)_{i\in\ovl
n}\,\Th^{-1}\big).
\end{align*}
As well,
$$
\Th\,\big[T_n + (-1)^n\ovl T_n\big]\big(A_i(x_i)_{i\in\ovl
n}\big)\Th^{-1} = \big[\ovl T_n + (-1)^nT_n\big]\big(\Th
A_i(x_i)_{i\in\ovl n}\Th^{-1}\big).
$$
Using these relations in the sum of \eqref{eq:18} and \eqref{eq:19}, we
indeed get \eqref{eq:21}:
\begin{align*}
&\big[T^\mathrm{sym}_n + (-1)^{n}\ovl T^\mathrm{sym}_n\big]
\big(A_1(x_1)\cdots\big) = \frac{1}{2N}
\sum_{l=0}^{N-1}\Bigl(\Theta^{2l} \,\big[T_n+(-1)^n\ovl T_n\big]
\big(\Th^{-2l}A_1(x_1)\Th^{2l}\cdots\big)\, \Theta^{-2l}
\\
&\quad +\Theta^{2l+1}\,\big[\ovl T_n+(-1)^n T_n\big]
\big(\Th^{-2l-1}A_1(x_1)\Th^{2l+1}\cdots\big)\,\Theta^{-(2l+1)}\Bigr)
\\
&\quad = \big[T_n + (-1)^n\ovl T_n\big]\big(A_1(x_1)\cdots\big).
\tag*{\qed}
\end{align*}
\hideqed
\end{proof}

\section{Conclusion}

The argument in reference \cite{colinamuyverde} fails to grapple with
the nitty-gritty of renormalization: it presumes that suitably
renormalized interacting fields exist such that Wightman's axioms 0 to
II are fulfilled. Instead a reasonable avenue is to concentrate just
on proving or disproving PCT invariance of the TOP for large enough
classes of models, constructed as rigorously as possible in
perturbation theory by dealing with renormalization through the causal
method.

We have precisely shown how PCT propagates through causal perturbative
renormalization. It pertains to point out caveats for our theorem. To
begin with, if someone ever were clever enough to come out with a
non-invariant $T_1$, there is nothing to do. We have focused on
the~TOP. For \textit{physical} PCT conservation to ensue, the
adiabatic limit of the model after renormalization must exist. As
pointed out by Kobayashi and Sanda time ago~\cite{banzai}, the remit
of any approach to PCT conservation is narrowed by the fact that both
heuristic and rigorous proofs make use of properties of asymptotic
states of particles that just do not apply in~QCD, whereupon the
elementary excitations of the field are confined. One has to admit
that large parts of the Standard Model disown the basic hypothesis of
any $\nS$-matrix theory.

Even for a garden-variety model like QED, the adiabatic limit
\cite{PhilippeRoland} exists only in a weak sense. In addition, as
conclusively shown by Herdegen, Dirac fields respecting Gauss' law
which are only ``spatially local'' provide the best tool to construct
QED, free from the infrared catastrophe~\cite{Herdegen}. We believe
that this should not decisively impinge on the issues considered
here~\cite{HerdegenII}; but the question is open to debate. It should
be stressed that dropping locality of the interactions, PCT
non-conservation and Lorentz invariance can perfectly
coexist~\cite{4mosqueteros,reckoning}. Now, the borderline between
local and non-local models is not nearly as neat as one would like. It
would appear that local and non-local fields may share the same
$\nS$-matrix~\cite{granego}, for that matter.

\appendix

\subsection*{Appendix. From Lorentz covariance to local commutativity}
\label{ssc:in-case}

To simplify the notation we consider a (possibly interacting) scalar
field $\phi(x)$.

\smallskip

\begin{hyp}
$\phi(x)$ is $L_+^\uparrow$-covariant, that is, there exists a
representation $L_+^\uparrow\ni\Lambda\mapsto U(\Lambda)$ such that
\begin{equation*}
\phi(\Lambda x) = U(\Lambda)\,\phi(x)\,U(\Lambda)^{-1}
\word{for all} \Lambda\in L_+^\uparrow.
\end{equation*}
In addition the TOP
\begin{equation*}
\label{0.2}
T(\phi(x)\,\phi(y)) := H(x^0 - y^0)\,\phi(x)\,\phi(y) + H(y^0 -
x^0)\,\phi(y)\,\phi(x) \word{for all} x^0 \ne y^0
\end{equation*}
(where $H$ denotes the Heaviside function) is $L_+^\uparrow$-covariant
\begin{equation}
T(\phi(\Lambda x)\,\phi(\Lambda y)) = U(\Lambda)\,T(\phi(x)\,\phi(y))
\,U(\Lambda)^{-1} \word{for all} x \ne y.
\label{eq:iam-senior-sed-cruda-deo-viridisque-senectus}
\end{equation}
In particular, assumption
\eqref{eq:iam-senior-sed-cruda-deo-viridisque-senectus} means that for
$x^0=y^0$ and $x\ne y$ the definition
\begin{equation*}
\label{0.4}
T(\phi(x)\,\phi(y)) := U(\Lambda)^{-1}\,T\big(\phi(\Lambda
x)\,\phi(\Lambda y)\big)\,U(\Lambda)
\end{equation*}
is independent of the choice of $\Lambda\in L_+^\uparrow$ with
$(\Lambda x)^0\ne(\Lambda y)^0$.
\end{hyp}

\smallskip

The following derivation of \textit{local commutativity}
\eqref{eq:0.6} is motivated by \cite{colinamuyverdeII}. Let $x,\,y$ be
given such that $(x-y)^2<0$. We choose $\Lambda_1,\Lambda_2 \in
L_+^\uparrow$ such that $(\Lambda_1 x)^0> (\Lambda_1 y)^0$ and
$(\Lambda_2 x)^0< (\Lambda_2 y)^0$. With that one obtains
\begin{align*}
\phi(x)\,\phi(y)& = U(\Lambda_1)^{-1}\,\phi(\Lambda_1
x)\,\phi(\Lambda_1 y)
\,U(\Lambda_1)=U(\Lambda_1)^{-1}\,T(\phi(\Lambda_1 x)\,\phi(\Lambda_1
y)) \,U(\Lambda_1)\notag\\
& =T(\phi(x)\,\phi(y))=
U(\Lambda_2)^{-1}\,T(\phi(\Lambda_2 x)\,\phi(\Lambda_2
y))\,U(\Lambda_2)\notag\\
& =U(\Lambda_2)^{-1}\,\phi(\Lambda_2 y)\,\phi(\Lambda_2 x)
\,U(\Lambda_2)=\phi(y)\,\phi(x). \qquad \square
\end{align*}

\subsection*{Acknowledgments}
\label{ssc:in-gratitude}

MD thanks the Departamento de F\'{\i}sica Te\'orica of the Universidad
de Zaragoza and the Instituto de Ciencias Matem\'aticas (CSIC--UAM,
Madrid) for their hospitality. We are deeply grateful to Masud
Chaichian and Anca Tureanu, who called the attention of JMG-B to
\cite{colinamuyverde}; to them, as well as to Philippe Blanchard,
Klaus Fredenhagen, Peter Pre\v{s}najder and Joseph C. V\'arilly for
very helpful discussions, suggestions and/or encouragement. Not least,
JMG-B is thankful to Gabriela~Barenboim, for prodding him on the
subject. He has been supported by grants FPA2009-09638 and
DGIID-DGA-E24/2, respectively of Spain's central and Aragon's regional
governments.

\end{document}